\documentclass[letterpaper, 10 pt, conference]{ieeeconf}
\IEEEoverridecommandlockouts
\usepackage{graphicx}      
\usepackage{subfig}
\usepackage{url}
\usepackage{hyperref}
\renewcommand{\baselinestretch}{0.9}

\usepackage{bm}

\usepackage{tikz-cd}
\usepackage{tikz}
\usepackage{circuitikz}
\usepackage{tikz-3dplot}
\tdplotsetmaincoords{60}{6}
\usetikzlibrary{3d}
\usetikzlibrary{arrows.meta, positioning, calc}
\usepackage{amsmath}
\allowdisplaybreaks
\usepackage[ruled,vlined]{algorithm2e}
\usepackage{algpseudocode}
\usepackage{amssymb}
\usepackage{mathtools}
\usepackage{amsfonts}
\usepackage{amsmath}
\usepackage{booktabs}
\usepackage{siunitx}
\usepackage{pifont}

\usepackage{mathrsfs}

\newtheorem{assumption}{Assumption}

\newtheorem{remark}{Remark}
\newtheorem{theorem}{Theorem}
\newtheorem{lemma}{Lemma}

\usepackage{comment}

\usepackage{caption}

\title{
Distributed State Estimation for Discrete-time LTI Systems: \\the Design Trilemma and a Novel Framework
}
\author{Ruixuan~Zhao, Guitao~Yang, James~Fleming, and Boli~Chen
\thanks{R. Zhao and B. Chen are with the Department of Electronic and Electrical Engineering, University College London, London, UK \tt\small(ruixuan.zhao.22@ucl.ac.uk; boli.chen@ucl.ac.uk).}
\thanks{G. Yang and J. Fleming are with the Wolfson School of Mechanical, Electrical and Manufacturing Engineering, Loughborough University, Loughborough, UK \tt\small(g.yang@lboro.ac.uk; j.fleming@lboro.ac.uk).}
}

\begin{document}

\maketitle

\begin{abstract}                
  With the advancement of IoT technologies and the rapid expansion of cyber-physical systems, there is increasing interest in distributed state estimation, where multiple sensors collaboratively monitor large-scale dynamic systems.
  Compared with its continuous-time counterpart, a discrete-time distributed observer faces greater challenges, as it cannot exploit high-gain mechanisms or instantaneous communication. Existing approaches depend on three tightly coupled factors: (i) system observability, (ii) communication frequency and dimension of the exchanged information, and (iii) network connectivity. However, the interdependence among these factors remains underexplored. This paper identifies a fundamental trilemma among these factors and introduces a general design framework that balances them through an iterative semidefinite programming approach. As such, the proposed method mitigates the restrictive assumptions present in existing works. The effectiveness and generality of the proposed approach are demonstrated through a simulation example.
\end{abstract}

\smallskip

\section{Introduction}
State estimation plays a crucial role in the field of control systems design, as it serves to monitor system status and facilitate decision-making \cite{zhao2020roles}.
However, as modern systems continue to scale up, sensor measurements are increasingly distributed across wide geographical areas, making accurate estimation of the overall system state a significant challenge. Fortunately, with the rapid development of IoT devices, it is feasible to leverage communication among sensor nodes to estimate the state in a joint manner. For many, a centralized estimation scheme might be the initial intuition, i.e., first selecting a ``central node'' that collects all measurements via communication, then applying conventional observer design to estimate the state at the central node, and finally disseminating the estimated state to the remaining nodes. Despite the simplicity of this concept, the design suffers from communication delays in the estimated state information, as well as the vulnerability of the scheme (e.g., if the central node becomes faulty) \cite{liu2021dynamic}. More recently, distributed observer design has gained attention, as it overcomes the limitations of centralized approaches: it enables all nodes to jointly reconstruct the system state \cite{yang2022large}. Furthermore, the increasing ubiquity of digital sensing, computation, and communication has made discrete-time systems more prevalent in industry, which has motivated numerous discrete-time distributed observers \cite{zhao2025multi,zhao2026DUIOGA}.

Building on the above motivations, this paper investigates the problem of distributed state estimation for discrete-time systems. We revisit the key assumptions underlying existing observer designs (observability, communication, and connectivity), analyze their respective roles, and formalize the inherent trade-offs among them. This unified perspective reveals the fundamental limitations embedded in current design paradigms. Motivated by these insights, we introduce a novel design framework that preserves conceptual simplicity while effectively relaxing the restrictive assumptions of prior approaches.

\paragraph*{Notation}
The sets of real, positive real, and natural numbers are denoted by $\mathbb{R}$, $\mathbb{R}_{>0}$, and $\mathbb{N}$, respectively. The identity matrix of size $n$ is denoted by $I_{n}$. The all-zeros matrix of size $n\times m$ is denoted by $\mathbf{0}_{n\times m}$, and simply by $\mathbf{0}$ when the dimensions are clear from context. $\mathbf{1}_n$ is defined as a $n$-dimensional vector with all $1$ elements. The symbol $\|\cdot\|$ denotes the standard Euclidean norm, and $\otimes$ denotes the Kronecker product. For a square matrix $P$, the notation $P\succ0$ (or $P\succeq0$) signifies that $P$ is positive definite (or positive semi-definite), and $\mathrm{Tr}(P)$ represents the trace of $P$. Given matrices $P_1,P_2,\ldots,P_N$, ${\rm diag}(P_1,P_2,\ldots,P_N)$ constructs a block diagonal matrix. Furthermore, ${\rm col}(P_1,P_2,\ldots,P_N)$ denotes the stacked matrix $[P_1^\top,P_2^\top,\cdots,P_N^\top]^\top$.

\section{Problem Statement}
The distributed state estimation problem for a discrete-time linear time-invariant (LTI) system is formulated as follows.
Consider the LTI dynamics
\begin{equation}\label{eq:sys}
    \begin{aligned}
        x(t+1) &= A x(t) + B u(t) \\
        y_i(t) &= C_ix(t), \, i= 1, 2, \dots, N
    \end{aligned}
\end{equation}
where $t\in\mathbb{N}$ denotes the discrete sampling step.
The vectors $x \in \mathbb R^{n_x}$ and $u \in \mathbb R^{n_u}$ represent the system state and input, respectively, while $y_i \in \mathbb R^{n_{y_i}}$ denotes the local measurement available at Node~$i$. The sensor nodes are interconnected through a directed communication graph $\mathcal{G} = (\mathbf{N}, \mathcal{E}, \mathcal{A})$. Here $\mathbf{N} = \{1, 2, \dots, N\}$ is the set of sensor nodes, $\mathcal{E} \subseteq \mathbf{N} \times \mathbf{N}$ is the set of directed communication links and $\mathcal{A} = [a_{ij}] \in \mathbb{R}^{N \times N}$ is the adjacency matrix. Here, $a_{ij} = 1$ if $(i, j) \in \mathcal{E}$ (indicating that node $j$ can receive information from node $i$), and $a_{ij} = 0$ otherwise. The neighbor set of node $i$ is denoted by $\mathcal{N}_i=\{j|(j,i)\in\mathcal{E}\}$. To maintain simplicity and focus within the scope of this paper, we assume that the input $u$ is fully accessible to all nodes.

The objective is to design a distributed observer scheme $\mathcal{O} = \{ \mathcal{O}_i \}_{i \in \mathbf N}$, $\forall i \in \mathbf{N}$, such that the system state $x(t)$ can be asymptotically estimated at each sensor node, 
\begin{equation}\label{eq:full-converge}
    \lim_{t \to +\infty} \lVert x(t) - \hat x_i(t) \lVert = 0, \ \forall i \in \mathbf{N}.
\end{equation}    
A typical structure of a distributed state estimation scheme is depicted in Fig.~\ref{fig:observerframework}.
\begin{figure}[htp!]
    \centering
    \scalebox{0.65}{\begin{tikzpicture}[tdplot_main_coords]
    \def\off{18}
    \def\N{5}
    \def\R{2.5}

    \pgfmathparse{360/\N}
    \edef\step{\pgfmathresult}

    \colorlet{net}{teal}
    \tikzset{comm/.style = {color=cyan, very thick, dash pattern=on 4pt off 1.5pt}}
    \foreach \x in {1,...,\N} {%
        \pgfmathparse{\R*cos(\x*\step+\off)}
        \let\xcoord\pgfmathresult
        \pgfmathparse{\R*sin(\x*\step+\off)}
        \let\ycoord\pgfmathresult
        \begin{scope}[canvas is xy plane at z=0,transform shape]
        \node [circle, 
            draw,
            color=white, 
            fill=white, 
            text=black, 
            very thick,
            inner sep=6pt] 
            (N\x) at (\xcoord, \ycoord, 0) {};
        \end{scope}
    }

\def\xO{0}
\def\yO{0}
\def\zO{-0.8}
\def\RR{2.5}
\def\HH{1}

\draw[densely dashed, domain=0:180, samples=37, variable=\t]
  plot ({\xO+\RR*cos(\t)}, {\yO+\RR*sin(\t)}, {\zO});

\draw[thick,
  preaction={draw=white, line width=2\pgflinewidth},
  domain=0:180, samples=37, variable=\t]
  plot ({\xO+\RR*cos(\t)}, {\yO+\RR*sin(\t)}, {\zO+\HH});

\draw[thick,
  preaction={draw=white, line width=2\pgflinewidth},
  domain=180:360, samples=37, variable=\t
] plot ({\xO+\RR*cos(\t)}, {\yO+\RR*sin(\t)}, {\zO});

\draw[thick,
  preaction={draw=white, line width=2\pgflinewidth},
  domain=180:360, samples=37, variable=\t
] plot ({\xO+\RR*cos(\t)}, {\yO+\RR*sin(\t)}, {\zO+\HH});

\draw[thick, domain=0:1, samples=2, variable=\s]
  plot ({\xO-\RR}, {\yO}, {\zO + \s*\HH});
\draw[thick, domain=0:1, samples=2, variable=\s]
  plot ({\xO+\RR}, {\yO}, {\zO + \s*\HH});
 
    \begin{scope}[canvas is xy plane at z=0,transform shape]
        \node[above, font=\Large] at (0,0,0) {LTI system};
    \end{scope}

    \foreach \x in {1,...,\N} {%
        \pgfmathparse{\R*cos(\x*\step+\off)}
        \let\xcoord\pgfmathresult
        \pgfmathparse{\R*sin(\x*\step+\off)}
        \let\ycoord\pgfmathresult
        \begin{scope}[canvas is xy plane at z=4,transform shape]
        \node [circle, 
            draw,
            color=net, 
            fill=white, 
            text=black, 
            very thick,
            inner sep=6pt] 
            (O\x) at (\xcoord, \ycoord, 4) {$\mathcal{O}_\x$};
        \end{scope}
    }
    \coordinate (D) at (0,0,4);
\pgfmathanglebetweenpoints{
    \pgfpointanchor{O1}{west}}{
    \pgfpointanchor{D}{center}}
\edef\anglestart{\pgfmathresult}
\pgfmathanglebetweenpoints{
    \pgfpointanchor{O2}{north}}{
    \pgfpointanchor{D}{center}}
\edef\angleend{\pgfmathresult-11}
\begin{scope}[canvas is xy plane at z=4,transform shape]
    \foreach \x in {1,...,\N} {
    \edef\rotation{\step*(\x-1)}
    \draw[comm,rotate=\rotation] (\anglestart-180:\R) arc[radius = \R, start angle = \anglestart-180, end angle = \angleend-180];
    }
\end{scope}

\draw [comm] (O2) -- (O5);
\draw [comm] (O1) -- (O3);

\draw [gray,-latex] (N1) -- (O1) node[midway, left, align=left] {$y_1$};
\draw [gray,-latex] (N2) -- (O2) node[near start, left, align=right] {$y_2$};
\draw [gray,-latex] (N3) -- (O3) node[midway, right, align=left] {$y_3$};
\draw [gray,-latex] (N4) -- (O4) node[midway, right, align=right] {$y_4$};
\draw [gray,-latex] (N5) -- (O5) node[near start, right, align=left] {$y_5$};

\end{tikzpicture} }\\[-1ex]
    \caption{Distributed observer framework with five nodes. Each observer $\mathcal{O}_i$ uses its local measurement $y_i$ and information from neighbors via the cyan dashed communication links.} \label{fig:observerframework}
\end{figure}
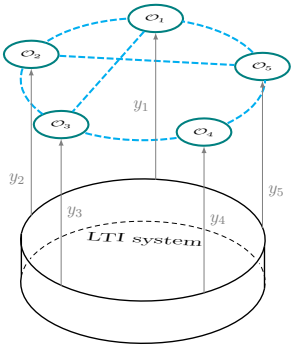

Given system \eqref{eq:sys}, the existence of a distributed estimator relies on three fundamental assumptions: joint detectability, information exchange, and graph connectivity. These assumptions are reviewed below.
\begin{assumption}[Marginally Joint Detectability] \label{assum:dect}
    The pair $(C,A)$ is detectable, where 
    \begin{equation*}
        C 
        = 
        \begin{bmatrix}
            C_1^\top & C_2^\top & \ldots & C_N^\top    
        \end{bmatrix}^\top.
    \end{equation*}
However, the pair $(\breve{C}, A)$ with $\breve{C} = \mathrm{col}_{i \in \mathbf{N} \setminus \{j\}}(C_i), \, \forall j \in \mathbf{N}$ is not detectable.
\end{assumption}
In contrast to the joint detectability condition commonly used in the literature, which only requires the detectability of $(C, A)$ and may include trivial cases with nearly detectable local nodes (e.g., $\exists i \in \mathbf{N}$ such that $(C_i, A)$ is detectable), we introduce the concept of \emph{Marginally Joint Detectability}, representing the minimal requirement for achieving joint state estimation.
\begin{assumption}[Information Exchange] \label{assum:trans}
    At each sampling instant $t$, every node communicates once and transmits only its latest local state estimate $\hat{x}_i(t)$.
\end{assumption}
%
This assumption implies that the communication frequency is identical to the system’s sampling frequency, and the dimension of the exchanged information remains unchanged. This represents the standard setting adopted in many existing distributed observer designs. In contrast, some studies assume that the dimension of transmitted information increases with the number of nodes~\cite{wang2018distributed}\footnote{Although a continuous-time system is considered in this paper, the proposed method's capability allows its direct application to both continuous- and discrete-time systems.}. A two-timescale communication framework, where multiple exchanges occur within a single sampling interval, is also common in distributed observer designs~\cite{accikmecse2014decentralized}. These configurations generally incur higher communication costs than Assumption~\ref{assum:trans}. It is also worth noting that, under a two-timescale setting, exact convergence in~\eqref{eq:full-converge} can typically be guaranteed as the communication frequency approaches infinity~\cite{yang2025state}.
%
\begin{assumption}[Graph Connectivity]\label{assum:connect}
    The communication graph $\mathcal{G}$ is a minimal strong digraph,  that is, $\mathcal{G}$ is strongly connected and, for every edge $\mathtt{e} \in \mathcal{E}$, the graph $\mathcal{G} \setminus \mathtt{e}$ is no longer strongly connected.
\end{assumption}
%
    Assumption~\ref{assum:connect} is clearly sufficient, as it ensures the propagation of information from every node to all others. However, the necessity of this assumption must be evaluated in conjunction with Assumption~\ref{assum:dect}, since, in a trivial case, a fully detectable node can reconstruct the state independently, even if it is disconnected from the rest of the network. As noted in~\cite{park2017design,mitra2018distributed}, strong connectivity of the entire sensor network is not always required; instead, it suffices that there exists a strongly connected subgraph that is jointly detectable or observable, a clear trade-off with Assumption~\ref{assum:dect}. Within such a subgraph, full state reconstruction can be achieved, while the remaining nodes can estimate the system state successfully, provided that this detectable/observable subgraph serves as their \textit{source component}.

The remainder of this paper is organized as follows: 
In Section~\ref{sec:review}, we examine the existing literature through the lens of these assumptions to highlight their implications and design trade-offs. Subsequently, in Section~\ref{sec:design}, we propose our own design, building upon these assumptions to address the identified trade-offs. Finally, we validate our design through a numerical example in Section~\ref{sec:simulate} and provide concluding remarks in Section~\ref{sec:conclude}.

\section{Literature Review}\label{sec:review}
Note that Assumptions~\ref{assum:dect}, \ref{assum:trans} and \ref{assum:connect} are intentionally stated at a ``minimal level''. A central observation from the literature is that designs that only request Assumptions~\ref {assum:dect}, \ref{assum:trans} and \ref{assum:connect} at their minimal levels simultaneously do not exist. Instead, successful methods typically require strengthening in at least one dimension: 
\begin{enumerate}
    \item they leverage a stronger local observability condition (e.g. each individual node is fully detectable/observable or each individual node combined with its first order in-neighbors is fully detectable/observable);
    \item they require richer communication (e.g., multiple information exchanges per time step or higher bandwidth);
    \item they assume stronger connectivity (e.g., a fully connected graph).
\end{enumerate}
We refer to this tension as a design trilemma (Fig.~\ref{fig:trilemma}): pushing one corner down generally forces one of the others up.
\begin{figure}[htp!]
\centering
\includegraphics[width=0.75\columnwidth]{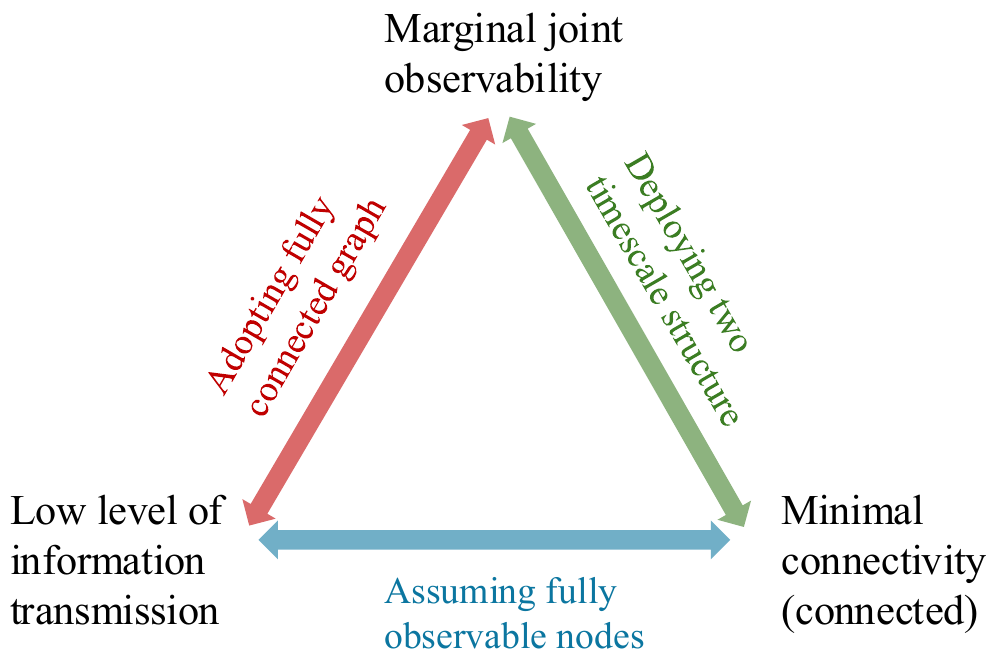}\\[-1.5ex]
\caption{Trilemma for discrete-time distributed state estimation on LTI plant with static graph.}
\label{fig:trilemma}
\end{figure}

Below, we summarize the major families of methods and explicitly state whether they invoke the assumptions at the minimal level or more restrictive versions. 

A seminal distributed Kalman filtering design was proposed by Olfati-Saber \cite{olfati2005distributed, olfati2007distributed, olfati2009kalman} and has been extended by many others (e.g., \cite{khan2011stability}). This type of design fuses local Kalman updates with consensus iterations and jointly ensures a lower variance level of the estimation error. These approaches typically violate Assumption~\ref{assum:trans} by adopting a two timescale structure: several communication rounds occur within a single sampling period to drive consensus among neighbors (Fig.~\ref{fig:timescale}). 
\begin{figure}[htp!]
\centering
\includegraphics[width=0.45\textwidth]{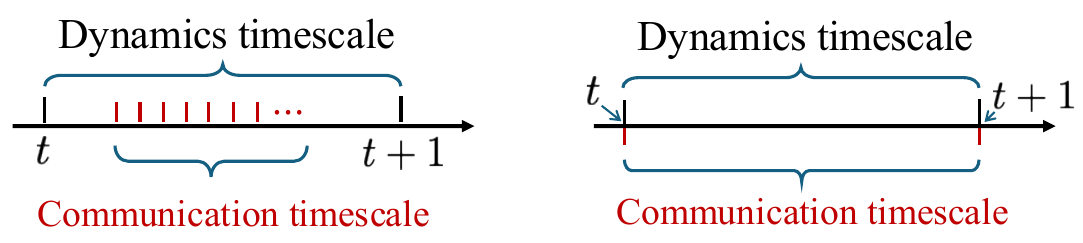}
\caption{(Left) Two timescale structure; (Right) Single timescale structure.}
\label{fig:timescale}
\end{figure}
While the extra exchanges improve fusion accuracy and can lower steady-state error covariance, the price is a higher communication frequency relative to Assumption~\ref{assum:trans}. Early variants also impose additional conditions beyond the minimal Assumption~\ref{assum:dect}. For example, the observer in \cite{olfati2005distributed} considers identical output matrices across all sensor nodes. Under such a condition, it is required that each individual node is fully observable to ensure the existence of a solution, which is much stronger than Assumption~\ref{assum:dect}.

A second stream develops Luenberger-type observers in which each node runs a local Luenberger update, coupled with inter-node correction terms. Within this family, there are two structural variants that sit on different edges of our trilemma.

Several designs embed multiple consensus iterations within each sampling period to drive agreement of estimates or innovations before the next plant update, e.g., \cite{accikmecse2014decentralized, wang2019distributed, rego2021distributed, liu2023distributed, wang2024split}. This architecture improves fusion quality without augmenting the plant state, but it violates Assumption~\ref{assum:trans} as it typically requires more than one information exchange per step. Many of these works relax the requirement that every individual sensor node should be fully observable in \cite{olfati2005distributed} because the extra exchange rounds can compensate for weaker local sensing requirements.

Other Luenberger-type approaches preserve Assumption~\ref{assum:trans} by allowing only one communication per sampling instant. Representative examples are the canonical decomposition based schemes in \cite{mitra2018distributed, mitra2022distributed}, which design local observers in coordinates aligned with the observable subspaces and couple node estimates through structured interconnections. Avoiding state augmentation and maintaining a single-round messaging exchange at the cost of stronger graph restrictions (communication flow consistent with the decomposition), which means Assumption~\ref{assum:connect} is strengthened beyond mere strong connectivity.

Another line of approaches augments the estimator states while retaining a single timescale, inspired by the decentralized control problem. In schemes such as \cite{wang2018distributed} and \cite{park2012augmented}, extra states are introduced both in the local iteration and in the communication variables at each step to ensure convergence. However, the number of augmented states typically scales up with the number of nodes, resulting in higher-dimensional information exchange, which increases the communication burden in large sensor networks.  In contrast, \cite{park2017design} also uses state-augmentation, and the total augmentation depends only on the number of source components, which are defined as strongly connected subgraphs that are jointly detectable/observable. Despite the augmentation, the dimension of information exchanged among agents remains equal to the system state, and it does not require multiple consensus iterations per sampling step (consistent with Assumption~\ref{assum:trans}). This advantage, however, comes at the expense of additional structural constraints on the communication topology and joint observability (Assumptions A1–A3 in \cite{park2017design}), alongside Assumptions~\ref{assum:dect} and~\ref{assum:connect}. 

With increasing computational power, optimization-based designs for distributed observers have attracted significant attention. In particular, distributed moving horizon estimation (MHE) has been widely studied because it naturally handles state constraints, supports robust analysis under diverse channel and measurement disturbances, and extends seamlessly to nonlinear systems. The formulation in \cite{farina2010distributed} introduces a distributed MHE that accounts for both process and measurement noise; however, it requires one-hop observability, which indicates that any nodes together with its in-degree neighboring nodes must already be observable rather than relying on observability across the entire network. In contrast, \cite{philipp2012distributed} proposed a distributed MHE applicable to an all-to-all communication architecture among estimators, which relaxes the observability requirement but clearly demands a more restrictive communication graph connectivity. Moreover, two timescale structures have been adopted in more recent distributed optimization-based estimation designs \cite{yang2025state, battistelli2018distributed, li2023iterative, pedroso2023decentralized}, where multiple communication/optimization iterations occur within each sampling period.

In summary, existing designs generally fail to satisfy Assumptions~\ref{assum:dect}, \ref{assum:trans}, and \ref{assum:connect} simultaneously at their minimal levels. In the successful approaches, at least one corner of the trilemma must be strengthened. Table~\ref{tab:literature} summarizes these trade-offs across existing methods. 
The conceptual trilemma highlights the fundamental limitation in current designs. Motivated by this insight, in the next section, we propose a novel single-timescale, non-augmented estimator that preserves state-dimension messages while trading off graph connectivity against local node observability. That is, we preserve Assumption~\ref{assum:trans}, while imposing a new condition that mixes observability and connectivity aspects. 
\begin{table}[htp!]
\begin{center}
    \begin{tabular}{SSSS} \toprule
        {Ref.} & {Observability} & {Communication} & {Connectivity} \\ \midrule
        \cite{olfati2005distributed}  & \ding{54} & \ding{51} & \ding{51}  \\
        \cite{olfati2007distributed} \cite{olfati2009kalman} \cite{khan2011stability}
        & \ding{51} & \ding{54} & \ding{51}   \\ \midrule
        \cite{accikmecse2014decentralized} 
        \cite{wang2019distributed} \cite{rego2021distributed} \cite{liu2023distributed}
        \cite{wang2024split}  & \ding{51} & \ding{54} & \ding{51}   \\
        \cite{mitra2018distributed}
        \cite{mitra2022distributed}  & \ding{51} & \ding{51} & \ding{54}   \\ \midrule
         \cite{wang2018distributed}
        \cite{park2012augmented}
        & \ding{51} & \ding{54} & \ding{51} \\ 
        *\cite{park2017design}
        & \ding{55}  & \ding{51} & \ding{55} \\\midrule
        \cite{farina2010distributed} & \ding{54} & \ding{51} & \ding{51} \\
        \cite{philipp2012distributed} & \ding{51} & \ding{51} & \ding{54} \\
        \cite{yang2025state}
        \cite{battistelli2018distributed} 
        \cite{li2023iterative} \cite{pedroso2023decentralized}  & \ding{51} & \ding{54} & \ding{51} \\ \midrule
        $\text{This paper}$ & \ding{55} & \ding{51} & \ding{55} \\\bottomrule
    \end{tabular}
\caption{
    Trilemma trade-offs of the existing literature on distributed state estimation for discrete-time LTI systems. Each entry indicates whether a method strengthens the requirements beyond the Assumptions~\ref{assum:dect}, \ref{assum:trans} and \ref{assum:connect}: \ding{51} = at/near minimal; \ding{54} = stricter than minimal; \ding{55} = trade-off with another assumption. The four horizontal blocks group corresponds to Kalman–consensus filters, Luenberger-type observers, augmented-state observers, and optimization/MHE-based designs, respectively. *\cite{park2017design}, similar to our approach, relies solely on Assumption~\ref{assum:trans} and allows a trade-off between observability and connectivity, though it tends to yield more conservative results, as will be demonstrated later. 
}
\label{tab:literature}
\end{center}
\end{table}

\section{Proposed Distributed Observer Design}\label{sec:design}
In this section, we introduce a novel framework for distributed observer design in discrete-time on a single timescale structure, as stated in Assumption~\ref{assum:trans}.

The distributed observer at Node $i$, $i\in\mathbf{N}$, is designed as
\begin{align}
        &\hat x_i(t+1) = A\left(\varpi_{ii}\hat x_i(t) + \sum_{j\in\mathcal{N}_i}\varpi_{ij}\hat{x}_j (t) \right) + Bu(t) \notag\\&\quad + L_i\big(C_i\hat{x}_i(t) - y_i(t)\big) + \sum_{j\in\mathcal{N}_i} M_{ij}\left(\hat{x}_j(t) - \hat x_i(t) \right)\label{eq:observer}
\end{align}
where $L_i\in\mathbb{R}^{n_x\times n_{y_i}}$, $M_{ij}\in\mathbb{R}^{n_x\times n_x}$, and $\varpi_{ij}\in\mathbb{R}$ with $\varpi_{ii}+\sum_{j\in\mathcal{N}_i}\varpi_{ij}=1$. Given an LTI system \eqref{eq:sys} and the underlying communication network $\mathcal{G}$, the following main theorem states the existence condition of the distributed observer $\{\mathcal{O}_i\}_{i\in\mathbf{N}}$.
\begin{theorem}\label{thm:main}
Consider an LTI system \eqref{eq:sys}, the communication among all sensor nodes is described by the graph $\mathcal{G}$ and follows Assumption~\ref{assum:trans}. For each node, the local observer \eqref{eq:observer} can reconstruct the entire system state $x(t)$ asymptotically as \eqref{eq:full-converge} if there exists a positive definite matrix $\bm Q\in\mathbb{R}^{Nn_x\times Nn_x}$ and matrices $\{\mathbf{W},\mathbf{L},\mathbf{M}\}$ such that
\begin{equation}\label{eq:nonlinearMI}
    \begin{bmatrix}
        -\bm Q&(\mathbf{W}\otimes A + \mathbf{L}\mathbf{C} - \mathbf{M})^\top\\
        \mathbf{W}\otimes A + \mathbf{L}\mathbf{C} - \mathbf{M}& -\bm Q^{-1}
    \end{bmatrix}\prec 0
\end{equation}
holds, where $\mathbf{W}=[\varpi_{ij}]_{i,j\in\mathbf{N}}\in\mathbb{R}^{N\times N}$ ($\varpi_{ij}=0$ if $(j,i)\notin\mathcal{E}$) with $\mathbf{W}\mathbf{1}_N=\mathbf{1}_N$, $\mathbf{L}=\mathrm{diag}(L_1,L_2,\cdots,L_N)\in\mathbb{R}^{Nn_x\times \sum_{i=1}^N n_{y_i}}$, $\mathbf{C}=\mathrm{diag}(C_1,C_2,\cdots,C_N)\in\mathbb{R}^{\sum_{i=1}^N n_{y_i}\times Nn_x}$, and $\mathbf{M}= [\breve{M}_{ij}]_{i,j\in\mathbf{N}}\in\mathbb{R}^{Nn_x\times Nn_x}$ with
\begin{equation*}
    \breve{M}_{ij}=\left\{\begin{aligned}
        \sum_{q\in\mathcal{N}_p}M_{pq}\quad &\text{if}\ i=j=p,\\
        -M_{ij}\quad &\text{if}\ (j,i)\in\mathcal{E},\\
        \mathbf{0}_{n\times n}\quad &\text{if}\ (j,i)\notin\mathcal{E}.
    \end{aligned}\right.
\end{equation*}
\end{theorem}
\begin{proof}
Define $e_i(t)\coloneq x(t)-\hat x_i(t)$ as the estimation error at Node $i$, then its dynamics can be written as follows
\begin{align}
        e_i(t+1) = &A\left(\varpi_{ii} e_i(t) + \sum_{j\in\mathcal{N}_i}\varpi_{ij}e_j(t)\right)+L_iC_i e_i(t)\notag\\& - \left( \sum_{j\in\mathcal{N}_i} M_{ij}\big(e_i(t) - e_j(t) \big) \right)
\end{align}
Stacking all local errors as $e\coloneq\mathrm{col}(e_1,e_2,\cdots,e_N)$, one can obtain the full error dynamics as
\begin{equation}\label{eq:error_dynamics}
    e(t+1)=(\mathbf{W}\otimes A + \mathbf{L}\mathbf{C} - \mathbf{M})e(t).
\end{equation}
Consider a Lyapunov candidate $V(t) \coloneq e(t)^\top \bm Q e(t)$.
The time difference of $V(t)$ along the error dynamics \eqref{eq:error_dynamics} follows
\begin{equation*}
    V(t+1) - V(t) = e(t)^\top \left(\mathbf{S}^\top \bm Q \mathbf{S} - \bm Q \right)e(t),
\end{equation*}
where $\mathbf{S}=\mathbf{W}\otimes A + \mathbf{L}\mathbf{C} - \mathbf{M}$. By applying the Schur Complement Lemma \cite{boyd1994linear}, the matrix inequality \eqref{eq:nonlinearMI} guarantees that
\begin{equation*}
    \mathbf{S}^\top \bm Q \mathbf{S} - \bm Q \prec 0,
\end{equation*}
which implies $V(t+1) - V(t)<0$ for all $e(t) \ne 0$. Consequently, the estimation error $e(t)$ converges asymptotically to zero.
\end{proof}

Due to the existence of a nonlinear term $\bm Q^{-1}$, it is very challenging to address the nonlinear matrix inequality \eqref{eq:nonlinearMI} directly. Therefore, we develop an iterative semi-definite programming (SDP) algorithm that transforms \eqref{eq:nonlinearMI} into a series of solvable optimization problems, which is inspired by the cone complementary linearization method \cite{mangasarian1995extended,el1997cone}.

Let $\bm P \coloneq \bm Q^{-1}$, the matrix inequality \eqref{eq:nonlinearMI} writes
    \begin{align}      
        &\begin{bmatrix}
        -\bm Q&\mathbf{S}^\top\\
        \mathbf{S}& -\bm P
    \end{bmatrix}\prec 0,\\
        &\bm P = \bm Q^{-1},\label{eq:nonconvex_constraint}\\ 
        &\mathbf{S}=\mathbf{W}\otimes A + \mathbf{L}\mathbf{C} - \mathbf{M}. \notag      
    \end{align}
Define the cone constraint
\begin{equation}\label{eq:cone_constraint}
    \mathcal{K}(\bm Q,\bm P)\coloneq\begin{bmatrix}
        \bm Q&I_{Nn_x}\\
        I_{Nn_x}&\bm P
    \end{bmatrix}\succeq 0
\end{equation}to be a convex relaxation of the nonconvex constraint \eqref{eq:nonconvex_constraint}. Subsequently, we can force the feasible solution to $\bm Q \bm P =I_{Nn_x}$, i.e. $\bm P = \bm Q^{-1}$, by minimizing $\mathrm{Tr}(\bm Q\bm P)$ through the following iterative SDP
\begin{equation}\label{eq:iter_opt}
        \mathscr{P}_k:\min_{\bm Q,\bm P, \mathbf{W}, \mathbf{L}, \mathbf{M}} \mathrm{Tr}(\bm{Q}_k \bm{P}) + \mathrm{Tr}(\bm{P}_k \bm{Q})
\end{equation}
subject to
\begin{equation}\label{eq:iter_LMI}
    \begin{aligned}
        & \begin{bmatrix}
            -\bm Q&\mathbf{S}^\top\\
        \mathbf{S}& -\bm P
        \end{bmatrix}\prec 0,\ 
        \mathbf{S}=\mathbf{W}\otimes A + \mathbf{L}\mathbf{C} - \mathbf{M},\\
        &\mathcal{K}(\bm Q,\bm P)\succeq 0,\ \bm Q\succ0,\ \mathbf{W}\mathbf{1}_N=\mathbf{1}_N.
    \end{aligned}
\end{equation}
In this optimization problem, $\mathrm{Tr}(\bm{Q}_k \bm{P}) + \mathrm{Tr}(\bm{P}_k \bm{Q})$ can be regarded as a first order symmetric linearization of $\mathrm{Tr}(\bm{Q}\bm{P})$ at $(\bm{Q}_k,\bm{P}_k)$.

\begin{algorithm}[ht]
\caption{Find matrix $\bm Q$, and calculate the gain matrices $\{\mathbf{W},\mathbf{L},\mathbf{M}\}$.}
\label{alg:cone}
\LinesNumbered
\KwIn{System matrix $A$, output matrices $C_i$, $i\in\mathbf{N}$, communication topology $\mathcal{G}$, and a small threshold $\varepsilon_{\mathrm{tol}}$.}
$k\leftarrow 0$\par
Solve the LMI problem \eqref{eq:iter_LMI}\textcolor{gray}{\Comment{Guarantee the initial value $\bm{Q}^\diamond$ and $\bm{P}^\diamond$ are in the feasible region $\mathcal{F}$.}}\par
$\bm Q_0\leftarrow \bm Q^\diamond$ and $\bm P_0\leftarrow \bm P^\diamond$\textcolor{gray}{\Comment{Get initial value.}}\par
\While{$\|\bm Q_k \bm P_k - I_{Nn_x}\|>\varepsilon_{\mathrm{tol}}$}{
Solve optimization problem $\mathscr{P}_k$ \eqref{eq:iter_opt}-\eqref{eq:iter_LMI}\par
$k\leftarrow k+1$\par
$\bm Q_k \leftarrow \bm Q$, $\bm P_k \leftarrow \bm P$\textcolor{gray}{\Comment{Update for next iteration.}}\par
\If{\eqref{eq:nonlinearMI} holds}{Break\textcolor{gray}{\Comment{Early stop if find solution of initial matrix inequality.}}}
}
\KwOut{Gain matrices $\{\mathbf{W},\mathbf{L},\mathbf{M}\}$ for distributed observer $\{\mathcal{O}_i\}_{i\in\mathbf{N}}$}
\end{algorithm}
The overall procedure is described in Algorithm~\ref{alg:cone}, while the convergence of the resulting iterative scheme is formalized in the lemma that follows.
\begin{lemma}
    Under Algorithm~\ref{alg:cone}, if the admissible set \begin{equation*}
        \mathcal{F}\coloneq\left\{(\bm Q,\bm P, \mathbf{W}, \mathbf{L}, \mathbf{M})\left|\begin{aligned}
            &\mathbf{S}=\mathbf{W}\otimes A + \mathbf{L}\mathbf{C} - \mathbf{M},\\
        &\begin{bmatrix}
            -\bm Q&\mathbf{S}^\top\\
        \mathbf{S}& -\bm P
        \end{bmatrix}\prec 0,\ \bm Q\succ0,\\
        &\mathcal{K}(\bm Q,\bm P)\succeq 0,\ \mathbf{W}\mathbf{1}_N=\mathbf{1}_N
        \end{aligned}\right.\right\}
    \end{equation*}
    is not empty, 
    we can define the sequence $J_k\coloneq\min_{(\bm Q,\bm P, \mathbf{W}, \mathbf{L}, \mathbf{M})\in\mathcal{F}}\ \mathrm{Tr}(\bm{Q}_k \bm{P}) + \mathrm{Tr}(\bm{P}_k\bm{Q})$ ($k\geq 0$), which has the following properties
    \begin{enumerate}
        \item $\{J_k\}$ is a monotonically non-increasing sequence, which converges to an optimum $J^\star$ with $J^\star\ge 2Nn_x$,
        \item If $J^\star = 2Nn_x$, every accumulation point $(\bm Q^\star,\bm P^\star, \mathbf{W}^\star, \mathbf{L}^\star, \mathbf{M}^\star)$ satisfies $\bm Q^\star\bm P^\star=I_{Nn_x}$,
        \item If there exists $(\widetilde{\bm Q},\widetilde{\bm P}, \widetilde{\mathbf{W}}, \widetilde{\mathbf{L}}, \widetilde{\mathbf{M}})\in\mathcal{F}$ with $\widetilde{\bm P}=\widetilde{\bm Q}^{-1}$ (complementary feasible point), then $J_k$ converges to $2Nn_x$ and any accumulation point satisfies $\bm Q^\star\bm P^\star=I_{Nn_x}$.
    \end{enumerate}
\end{lemma}
\begin{proof}
    Each optimization problem of Algorithm~\ref{alg:cone} is an SDP with a linear objective function over the convex set $\mathcal{F}$ satisfying Slater's condition. Thus, there exists a minimizer in $\mathcal{F}$.

    (Monotonicity) By definition, $(\bm Q_{k+1},\bm P_{k+1})$ is an optimal solution of $\mathscr{P}_k$. Hence, we have
    \begin{align*}
        J_k = &\mathrm{Tr}(\bm{Q}_k \bm{P}_{k+1}) + \mathrm{Tr}(\bm{P}_k\bm{Q}_{k+1})\\
        &\quad\leq \mathrm{Tr}(\bm{Q}_{k-1} \bm{P}_{k}) + \mathrm{Tr}(\bm{P}_{k-1}\bm{Q}_{k})=J_{k-1},
    \end{align*}
    hence $\{J_k\}$ is monotonically non-increasing.

    (Lower bound of convergence) For any fixed $(\bm Q_k,\bm P_k)$ and any $(\bm Q,\bm P)\in\mathcal{F}$, we have
    \begin{align*}
        \underset{\mathcal{K}(\bm Q,\bm P)\succeq 0}{\mathrm{inf}} \mathrm{Tr}(\bm{Q}_k \bm{P}) &+ \mathrm{Tr}(\bm{P}_k\bm{Q})\\ &= \underset{\bm P\succ 0}{\mathrm{inf}}\ \mathrm{Tr}(\bm{Q}_k \bm{P}) + \mathrm{Tr}(\bm{P}_k\bm{P}^{-1}),
    \end{align*}
    where the infimum is obtained at $\bm Q = \bm{P}^{-1}$.

    Define $f(\bm P) = \mathrm{Tr}(\bm{Q}_k \bm{P}) + \mathrm{Tr}(\bm{P}_k\bm{P}^{-1})$. Calculating the derivative of this function, one can obtain
    \begin{equation*}
        \nabla f(\bm P) = \bm Q_k -\bm P^{-1} \bm P_k \bm P^{-1}.
    \end{equation*}
    In accordance with the optimality condition $ \nabla f(\bm P)=\mathbf{0}$, we can get
    \begin{equation*}
        \bm Q_k = \bm P^{-1} \bm P_k \bm P^{-1} \Longleftrightarrow  \bm P \bm Q_k \bm P = \bm P_k,
    \end{equation*}
    which has a unique solution
    \begin{equation*}
        \bm P^\flat = \bm P_k^{\frac{1}{2}} \left(\bm P_k^{\frac{1}{2}}\bm Q_k\bm P_k^{\frac{1}{2}} \right)^{-\frac{1}{2}} \bm P_k^{\frac{1}{2}}.
    \end{equation*}
 Substituting it into the function $f(\bm P)$, we can obtain
    \begin{align*}
        \underset{\mathcal{K}(\bm Q,\bm P)\succeq 0}{\mathrm{inf}} &\mathrm{Tr}(\bm{Q}_k \bm{P}) + \mathrm{Tr}(\bm{P}_k\bm{Q}) \\&= 
        \underset{\bm P\succ 0}{\mathrm{inf}}\ f(\bm P) = 2\mathrm{Tr}\left(\left(\bm P_k^{\frac{1}{2}}\bm Q_k\bm P_k^{\frac{1}{2}} \right)^{\frac{1}{2}}\right).
    \end{align*}
    Since $\mathcal{K}(\bm Q,\bm P)\succeq 0$ holds, Schur Complement implies that
    $\bm P_k^{\frac{1}{2}}\bm Q_k\bm P_k^{\frac{1}{2}}\succeq I_{Nn_x}$, which leads to the following result
    \begin{equation}\label{eq:inf_QP}
         \mathrm{Tr}(\bm{Q}_k \bm{P}) + \mathrm{Tr}(\bm{P}_k\bm{Q}) \geq 2 Nn_x.
    \end{equation}
    Minimizing over $(\bm Q,\bm P)$ results in $J_k\geq 2Nn_x$, which proves $\lim_{k\rightarrow \infty}J_k=J^\star\geq 2 Nn_x$ together with the monotonicity.

    (Equality property) If $\lim_{k\rightarrow \infty}J_k=2 Nn_x$, then equality must hold in both \eqref{eq:inf_QP} and $\bm P_k^{\frac{1}{2}}\bm Q_k\bm P_k^{\frac{1}{2}}\succeq I_{Nn_x}$. $\bm P_k^{\frac{1}{2}}\bm Q_k\bm P_k^{\frac{1}{2}} = I_{Nn_x}$ holds if and only if $\bm Q_k \bm P_k=I_{Nn_x}$. Passing to any convergent subsequence yields $\bm Q^\star\bm P^\star=I_{Nn_x}$.

    (Complementarity feasibility) If there exists a feasible pair  $(\widetilde{\bm Q},\widetilde{\bm P})\in\mathcal{F}$ satisfying $\widetilde{\bm P}=\widetilde{\bm Q}^{-1}$, then for any other feasible pair  $(\hat{\bm Q},\hat{\bm P})$, $\mathrm{Tr}(\hat{\bm Q} \widetilde{\bm P})+\mathrm{Tr}(\hat{\bm P}\widetilde{\bm Q})\geq 2Nn_x$, with equality achieved when $\hat{\bm Q}=\widetilde{\bm P}^{-1}$, $\hat{\bm P}=\widetilde{\bm Q}^{-1}$. Therefore, all optimal values satisfy $J_k\geq2Nn_x$, and this lower bound is attainable. Combining this result with the monotonicity of $\{J_k\}$, we can obtain $\lim_{k\rightarrow\infty}J_k=2Nn_x$. According to the ``Equality property'', it follows that every accumulation point satisfies $\bm Q^\star\bm P^\star=I_{Nn_x}$.
\end{proof}
\begin{remark}\label{rem:1}
Theorem \ref{thm:main} establishes that the gain matrices ${\mathbf{W},\mathbf{L},\mathbf{M}}$ for the proposed distributed observer \eqref{eq:observer} exist if the matrix inequality \eqref{eq:nonlinearMI} is satisfied (sufficiency). When this condition holds, the matrices can be computed via Algorithm~\ref{alg:cone}. The inequality condition \eqref{eq:nonlinearMI} inherently couples the requirements for observability and connectivity. Compared with \cite{park2017design}, the key innovation lies in the newly introduced Laplacian-type coupling $M_{ij}$, which introduces an additional degree of freedom to enhance the contraction of non-average network modes beyond $\mathbf{W}\otimes A + \mathbf{L}\mathbf{C}$. This flexibility is particularly beneficial when the plant dynamics (i.e., $A$) are ``fast,'' thereby extending the applicability boundary of \cite{park2017design}, as illustrated by the simulation example in Section~\ref{sec:simulate}.  
\end{remark}

\section{Simulation Examples}\label{sec:simulate}
To validate the proposed framework and the trilemma, we consider a discrete-time LTI system of the form~\eqref{eq:sys} with
\begin{align}
    &A=\begin{bmatrix}
-1.5 & 0 & 0 & 0 & 0 \\ 
0 & 10 & 0 & 0 & 0 \\ 
0 & 0 & -2 & 0 & 0 \\ 
0 & 0 & 0 & -3 & 0 \\ 
0 & 0 & 0 & 0 & -2
\end{bmatrix},\ B=\begin{bmatrix}
1 & 0 \\ 
0 & 1 \\ 
0 & 1 \\ 
1 & 0 \\ 
1 & 0
\end{bmatrix}.\label{eq:simu_AB}
\end{align}
Five sensing nodes are deployed along the system, with output matrices given by
\begin{equation}
\begin{aligned}
&C_1= \begin{bmatrix}
1 & 0 & 0 & 0 & 0
\end{bmatrix},\ C_2=\begin{bmatrix}
0 & 1 & 0 & 0 & 0
\end{bmatrix},\\
&C_3 = \begin{bmatrix}
0 & 0 & 1 & 0 & 0
\end{bmatrix},\ C_4= \begin{bmatrix}
0 & 0 & 0 & 1 & 0
\end{bmatrix},\\
&C_5=\begin{bmatrix}
0 & 0 & 0 & 0 & 1
\end{bmatrix}.
\end{aligned}
\label{eq:simu_C}
\end{equation}
so that the overall pair $(C,A)$ is marginally jointly detectable (see Assumption~\ref{assum:dect}).
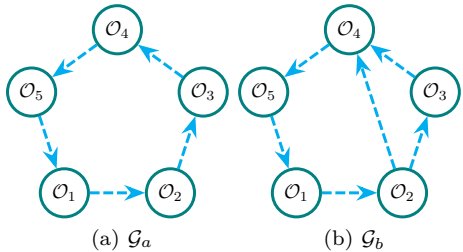
\begin{figure}[htp!]
    \centering
    \subfloat[$\mathcal{G}_a$]{\scalebox{0.94}{\begin{tikzpicture}
        \begin{scope}[every node/.style={scale=0.8,circle,very thick,draw,color=teal, text=black}] 
            \node (1) at (0,0) {$\mathcal O_1$};
            \node (2) at (1.5,0) {$\mathcal O_2$};
            \node (3) at (1.9635,1.4266) {$\mathcal O_3$};
            \node (4) at (0.7500,2.3083) {$\mathcal O_4$};
            \node (5) at (-0.4635,1.4266) {$\mathcal O_5$};
        \end{scope}
        \begin{scope}[>={Stealth[cyan]},
                      every edge/.style={draw=cyan, very thick, dash pattern=on 4pt off 1.5pt}] 
        \path [->] (1) edge node {} (2);
        \path [->] (2) edge node {} (3);
        \path [->] (3) edge node {} (4);
        \path [->] (4) edge node {} (5);
        \path [->] (5) edge node {} (1);
        \end{scope}
    \end{tikzpicture}}\label{fig:topo_a}}
    \subfloat[$\mathcal{G}_b$]{\scalebox{0.94}{\begin{tikzpicture}
        \begin{scope}[every node/.style={scale=0.8,circle,very thick,draw,color=teal, text=black}] 
            \node (1) at (0,0) {$\mathcal O_1$};
            \node (2) at (1.5,0) {$\mathcal O_2$};
            \node (3) at (1.9635,1.4266) {$\mathcal O_3$};
            \node (4) at (0.7500,2.3083) {$\mathcal O_4$};
            \node (5) at (-0.4635,1.4266) {$\mathcal O_5$};
        \end{scope}
        \begin{scope}[>={Stealth[cyan]},
                      every edge/.style={draw=cyan, very thick, dash pattern=on 4pt off 1.5pt}] 
        \path [->] (1) edge node {} (2);
        \path [->] (2) edge node {} (3);
        \path [->] (3) edge node {} (4);
        \path [->] (2) edge node {} (4);
        \path [->] (4) edge node {} (5);
        \path [->] (5) edge node {} (1);
        \end{scope}
    \end{tikzpicture}}\label{fig:topo_b}}\\[-1ex]
    \caption{Communication topologies for simulation examples}
    \label{fig:simu_topo}
\end{figure}
The initial communication graph $\mathcal{G}_a$ is a minimal strongly connected digraph (Fig.~\ref{fig:topo_a}), corresponding to the lowest connectivity level. Under this setting, Algorithm~\ref{alg:cone} fails to yield feasible gain matrices $\{\mathbf{W},\mathbf{L},\mathbf{M}\}$ satisfying~\eqref{eq:nonlinearMI}. According to the trilemma, to recover feasibility, we can strengthen the sensing at Node 4 with
\begin{equation*}
    C_4^{\prime} = \begin{bmatrix}
0 & 0 & 0 & 1 & 0 \\ 
0 & 1 & 0 & 0 & 0
\end{bmatrix}.
\end{equation*}
With this improvement, Algorithm~\ref{alg:cone} successfully yields feasible gains $\{\mathbf{W},\mathbf{L},\mathbf{M}\}$ (provided in the supplementary document\footnote{\label{docm:parameter}\href{https://github.com/RuixuanZhaoEEEUCL/ECC2026.git}{https://github.com/RuixuanZhaoEEEUCL/ECC2026.git}} due to space limitations), and all resulting estimation errors asymptotically vanish (Fig.~\ref{fig:error_a}). This confirms that improved observability alone can recover feasibility.
\begin{figure}[htp!]
    \centering
    \subfloat[Enhance observability]{\includegraphics[width=0.24\textwidth]{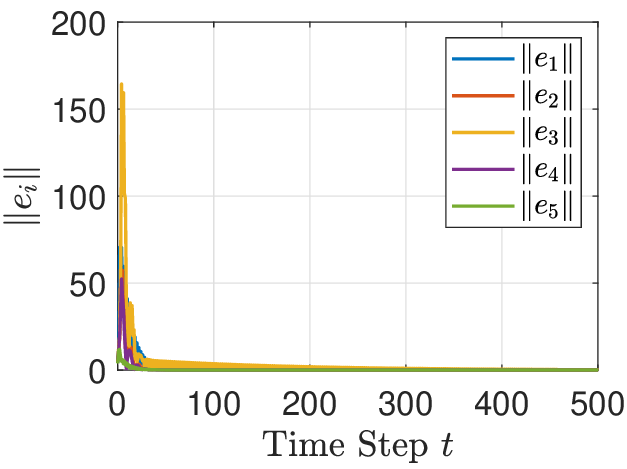}\label{fig:error_a}}
    \subfloat[Enhance connectivity]{\includegraphics[width=0.24\textwidth]{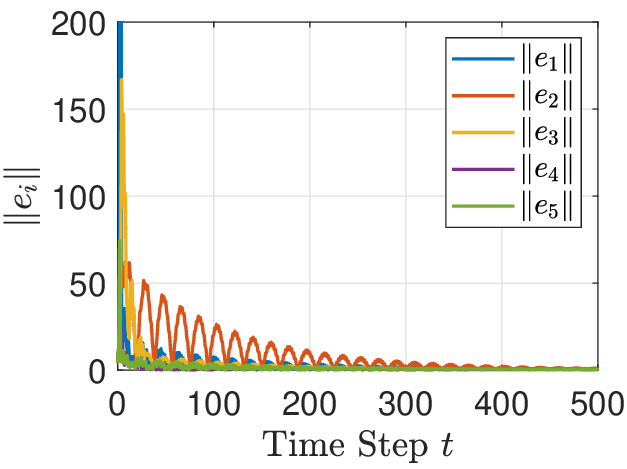}\label{fig:error_b}}\\[-1ex]
    \caption{Norm of estimation errors at each node for (a) enhanced observability, and (b) enhanced network connectivity.}
    \label{fig:estimation_error}
\end{figure}

Alternatively, we preserve the original sensors but add one directed edge to obtain $\mathcal{G}_b$, as shown in Fig.~\ref{fig:topo_b}. The increased connectivity  restores feasibility without modifying any $C_i$, feasible gain matrices $\{\mathbf{W},\mathbf{L},\mathbf{M}\}$ can be obtained by Algorithm~\ref{alg:cone}, (also provided in the supplementary document\textsuperscript{\ref{docm:parameter}}), and the resulting observer can ensure asymptotic convergence of the estimation error, as illustrated in Fig.~\ref{fig:error_b}.

Taken together, the two simulations demonstrate that feasibility can be restored by enhancing either observability or network connectivity (while keeping the information exchange pattern fixed), confirming the trilemma and the implicit coupling captured in Theorem~\ref{thm:main}.

Moreover, it can be observed that the same system described by \eqref{eq:simu_AB}--\eqref{eq:simu_C} under topology $\mathcal{G}_b$ cannot be handled by \cite{park2017design}, as the required design conditions in those works are violated. This comparison clearly demonstrates the broader generality and enhanced applicability of the proposed scheme.

\section{Conclusions and Future Works}\label{sec:conclude}
Building upon existing research in discrete-time distributed observers, this paper identifies a fundamental trilemma that captures the intrinsic trade-off among observability, connectivity, and communication frequency. To alleviate the stringent conditions commonly imposed in prior works, we propose a unified single-timescale design framework that inherently couples observability and connectivity requirements, leading to a nonlinear matrix inequality (NMI) formulation.  
The resulting NMI was effectively addressed through an iterative semidefinite programming (SDP) approach inspired by cone complementarity linearization. Simulation results verified the validity and enhanced the generality of the proposed design compared with existing methods.  
Future research will focus on establishing explicit analytical relationships within the identified trilemma and examining the necessary conditions of the design. 


\end{document}